\documentclass[conference,letterpaper]{IEEEtran}

\usepackage[utf8]{inputenc}
\usepackage{amsmath,amssymb,amsfonts}
\usepackage{mathtools}
\usepackage{amsthm}
\usepackage{algorithmic}
\usepackage{graphicx}
\usepackage{textcomp}
\usepackage{tikz}
\usepackage{caption}
\usepackage{cuted}
\usepackage{pgfgantt}
\usepackage{pdflscape}
\usepackage{pst-plot}
\usepackage{comment} 
\usepackage{cases}
\usepackage{nccfoots}
\usepackage{lineno,hyperref}
\usetikzlibrary{spy}
\usetikzlibrary{positioning,calc}
\usetikzlibrary{decorations.pathmorphing,calc,shapes,shapes.geometric,patterns}
\usetikzlibrary{shapes.multipart}
\usepackage{xfrac}
\usepackage{colortbl}
\usepackage{cancel} 
\usetikzlibrary{arrows,positioning,calc,intersections}
\usetikzlibrary{datavisualization.formats.functions}
\def\BibTeX{{\rm B\kern-.05em{\sc i\kern-.025em b}\kern-.08em
    T\kern-.1667em\lower.7ex\hbox{E}\kern-.125emX}}
    
\usepackage{romannum}
\usepackage{pgfplots}
\usepgfplotslibrary{fillbetween}
\usetikzlibrary{arrows, decorations.markings}

\newtheorem{theorem}{Theorem}
\newtheorem*{theorem*}{Theorem}
\newtheorem{lemma}[theorem]{Lemma}

\newtheorem{definition}[theorem]{Definition}

\newtheorem{remark}[theorem]{Remark}

\newcommand{\seta}{\ensuremath{\mathcal{A}}}
\newcommand{\setx}{\ensuremath{\mathcal{X}}}
\newcommand{\sety}{\ensuremath{\mathcal{Y}}}
\newcommand{\setz}{\ensuremath{\mathcal{Z}}}

\newcommand{\setu}{\ensuremath{\mathcal{U}}}
\newcommand{\sete}{\ensuremath{\mathcal{E}}}

\newcommand{\sett}{\ensuremath{\mathcal{T}}}

\newcommand{\circlearrow}{}
\DeclareRobustCommand{\circlearrow}{%
  \mathrel{\vphantom{\rightarrow}\mathpalette\circle@arrow\relax}%
}
\newcommand{\circle@arrow}[2]{%
  \m@th
  \ooalign{%
    \hidewidth$#1\circ\mkern1mu$\hidewidth\cr
    $#1-$\cr}%
}

\DeclarePairedDelimiterX{\infdivx}[2]{(}{)}{%
  #1\;\delimsize\|\;#2%
}

\begin{document}
\title{Common Randomness Generation from Gaussian Sources} 



\author{
\IEEEauthorblockN{Wafa Labidi\IEEEauthorrefmark{1}, Rami Ezzine\IEEEauthorrefmark{1}, Christian Deppe \IEEEauthorrefmark{2} and Holger Boche\IEEEauthorrefmark{1}\IEEEauthorrefmark{3}\IEEEauthorrefmark{4}}\\
\IEEEauthorblockA{\IEEEauthorrefmark{1}Technical University of Munich, Chair of Theoretical Information Technology, Munich, Germany\\
\IEEEauthorrefmark{2}Technical University of Munich, Institute for Communications Engineering, Munich, Germany \\
\IEEEauthorrefmark{3}CASA -- Cyber Security in the Age of Large-Scale Adversaries–
Exzellenzcluster, Ruhr-Universit\"at Bochum, Germany\\
\IEEEauthorrefmark{4}Munich Center for Quantum Science and Technology (MCQST), Schellingstr. 4, 80799 Munich, Germany\\
Email: \{wafa.labidi, rami.ezzine, christian.deppe, boche\}@tum.de}}

\maketitle


\begin{abstract}
We study the problem of common randomness (CR) generation in the basic two-party communication setting in which the sender and the receiver aim to agree on a common random variable with high probability by observing independent and identically distributed (i.i.d.)  samples of correlated Gaussian sources and while communicating as little as possible over a noisy memoryless channel. We completely solve the problem by giving a single-letter characterization of the CR capacity for the proposed model and by providing a rigorous proof of it. Interestingly, we prove that the CR capacity is infinite when the Gaussian sources are perfectly correlated. 
\end{abstract}
\section{Introduction}
In the context of common randomness (CR) generation, the sender and the receiver, often described as terminals, aim to agree on a common random variable with high probability.
The availability of this CR is advantageous as it allows to implement correlated random protocols that often perform faster and more efficiently than the deterministic ones or the ones using independent randomization.

 An enormous performance gain can be achieved by taking advantage of the resource CR in the identification scheme, since it may allow a significant increase in the identification capacity of channels \cite{trafo,part2,ahlswede2021}.
 The identification scheme is a new approach  in communications
developed by Ahlswede and Dueck \cite{Idchannels} in 1989. For many new applications with high requirements on reliability and latency such as  several machine-to-machine
and human-to-machine systems \cite{application}, the tactile internet
\cite{Tactilesinternet}, digital watermarking \cite{MOULINwatermarking,AhlswedeWatermarking,SteinbergWatermarking}, industry 4.0 \cite{industry4.0}, the identification approach is much more efficient than the classical transmission scheme  proposed by Shannon \cite{Shannon}. In the identification framework, the encoder sends an
identification message (called also identity) over the channel
and the decoder is not interested in what the received message
is, but wants to know whether a specific message has been
sent or not.

Many researches explored the problem of CR generation from correlated discrete sources. This problem was initially introduced by 
Ahlswede and Csizár in \cite{part2}, where the sender and the receiver are  additionally allowed to communicate over a  discrete noiseless channel with limited capacity. Unlike in the fundamental two papers \cite{part1}\cite{maurer}, no secrecy requirements are imposed.
A single-letter characterization of the CR capacity for that model was established in \cite{part2}. CR capacity refers to the maximum rate of CR that Alice and Bob can generate using the resources available in the model. 
Later, the results on CR capacity have been extended in \cite{globecom} to point-to-point single-input single-output (SISO) and  Multiple-Input Multiple-Output (MIMO) Gaussian channels  for  their  practical  relevance  in many communication situations such as  wired  and  wireless communications,  satellite  and  deep  space  communication  links, etc. The results on CR capacity over Gaussian channels have been used to establish a lower-bound on their corresponding correlation-assisted  secure identification capacity in the log-log scale \cite{globecom}. This lower bound can already exceed the secure identification capacity over Gaussian channels with randomized encoding elaborated in \cite{wafapaper}.
The problem of CR generation over SISO and MIMO fading channels has been investigated in  \cite{SISOfasingCR} and in \cite{MIMOfadingCR}, respectively, where the authors introduced the concept of outage in the CR generation framework.

However, as far as we know, there are no results regarding CR generation from correlated continuous sources. The main contribution of our work lies in establishing a single-letter characterization of the CR capacity for a model involving a bivariate Gaussian source with unidirectional communication over noisy memoryless channels. We will extend the CR capacity formula established in \cite{part2} for correlated discrete sources to correlated Gaussian sources. 
Interestingly, in contrast to the discrete case where the CR capacity is always finite \cite{part2}\cite{globecom}, we will show that the CR capacity is infinite when the Gaussian sources are perfectly correlated. In such a situation, no communication over the channel is required. We were motivated by the drastic effects on the identification capacity produced by the common randomness generated from the perfect feedback in the model treated in \cite{isit_paper}. The authors in \cite{isit_paper} proved that the identification capacity of Gaussian channels with noiseless feedback is infinite regardless of the scaling by proposing a coding scheme that generates an infinitely large amount of CR between the sender and the receiver using noiseless feedback.

Applications of our work include the problem of correlation-assisted identification, where the sender and the receiver have access to a correlated Gaussian source. Indeed, analogously to the discrete case \cite{globecom} and based on an early work in \cite{concatenation}, one can  construct identification codes for noisy memoryless channels based on the concatenation of two transmission codes 
using CR as a resource. 

\quad \textit{ Paper Outline:} The rest of the paper is organized as follows. In Section \ref{preliminaries}, we  introduce a generalized typicality criteria that can be applied to any i.i.d. continuous sources and we establish the conditional typicality lemma and conditional divergence lemma for the proposed typicality criteria using the weak law of large numbers (WLLN). In Section \ref{systemmodelanddefinitions}, 
we present the  system model for CR generation, provide the key definitions and the main result.
In  Section \ref{direct}, we will prove the achievability of the CR capacity by proposing a coding scheme based on the same type of binning as in the Wyner-Ziv problem, where we make use of the conditional typicality and the conditional divergence lemma elaborated in Section \ref{preliminaries}.
The converse proof of the CR capacity is established in Section \ref{converse}.
Section \ref{conclusion} contains concluding remarks.
\section{Preliminaries}
\label{preliminaries}
\subsection{Notations}
Calligraphic letters $\setx, \sety,\setz, \ldots$ are used for finite or infinite sets; lowercase letters $x,y,z,\ldots$ stand for constants and values of random variables; uppercase letters $X,Y,Z,\ldots$ stand for random variables;
 For any random variables $X$, $Y$ and $Z$, we use the notation $\color{black}X \circlearrow{Y} \circlearrow{Z}\color{black}$ to indicate a Markov chain.
$\mathbb{R}$ denotes the sets of real numbers; 
$p_X$ denotes the probability density function of a continuous RV $X$; $|\setx|$ denotes the cardinality a finite set $\setx$; the set of probability distributions on the set $\setx$ is denoted by $\mathcal{P}(\setx)$; $H(\cdot)$, $\mathbb{E}(\cdot)$ and $I(\cdot ;\cdot)$ are the entropy, the expected value and the mutual information, respectively;
all logarithms and information quantities are taken to base $2$.
\subsection{Typicality Criteria for Continuous Alphabet} \label{subsec: typicality}
Inspired by the generalized typicality criteria introduced in \cite{Generalized_typ} and based on the information-spectrum approach \cite{HanBook}, we define the following typicality criterion. This criterion can be applied to i.i.d. source/channel coding problems.
\begin{definition}
Suppose $\delta>0$ and $(X^n,Y^n)$ was emitted by the bivariate Gaussian memoryless source $P_{XY}$.
The sequence pair $(x^n, y^n)$ is called jointly $\delta$-typical with respect to $p_{XY}$ if
\begin{align}
    &|\frac{1}{n}i_{X^nY^n}(x^n,y^n)-I(X;Y)|\leq \delta, \quad \delta>0,
    \label{eq:criterion}
\end{align}
where $i_{X^nY^n}(x^n,y^n)$ is the information density \cite{HanBook} defined as 
\begin{equation*}
    i_{X^nY^n}(x^n,y^n)=\log\left( \frac{dp_{Y^n|X^n}(y^n|x^n)}{dp_{Y^n}(y^n)} \right)
\end{equation*}
when $p_{Y^n|X^n}$ is absolutely continuous w.r.t. $p_{Y^n}$. Let $\sett_\delta^{X^nY^n}$ denote the set of all $\delta$-jointly typical sequences.
\end{definition}
\begin{remark}
In \cite{mitran_polish} and \cite{Borel_Raginsky}, the authors introduced typicality criteria for measures on a Polish space and a Borel space, respectively. In \cite{Generalized_typ}, the authors considered only measurable spaces as alphabets.
\end{remark}
In the following, we consider the properties of sets with probability approaching one \cite{Generalized_typ}.
\begin{lemma}{\cite{Generalized_typ}}
Given a bivariate Gaussian memoryless source $p_{XY}$, we denote $\{\seta^{X^n Y^n}\}_{n=1}^{\infty}$ as a sequence of sets satisfying the following condition
\begin{equation}
    p_{X^nY^n}(\seta^{X^nY^n})\geq1-\alpha(n),\quad \lim_{n\to \infty} \alpha(n)=0
    \label{eq:jointTypicality}
\end{equation}
where $\seta^{X^n Y^n} \subset \setx^n \times \sety^n$ is $p_{X^nY^n}$-measurable for all $n\in \mathbb{N}$. Let
\begin{align*}
&\seta^{Y^n|x^n}= \{y^n \in \sety^n|(x^n, y^n) \in \seta^{X^n Y^n} \} \\
& \text{and }\seta^{X^n|Y^n}= \{x^n \in \setx^n|p_{Y^n|X^n}(\seta^{Y^n|x^n}|x^n) > 0\}.
\end{align*}
Then, for all $n\in \mathbb{N}$, the set $\{\seta^{X^n Y^n}\}_{n=1}^{\infty}$ has the following properties
\begin{align}
    & \lim_{n\to \infty} p_{X^n}(\seta^{X^n|Y^n})=1; \\
    &    \lim_{n\to \infty} p_{Y^n|X^n}(\seta^{Y^n|x^n}|x^n)=1, \quad \forall x^n \in \seta^{X^n|Y^n}.
\end{align}
\label{Lemma_typicalityCriterion}
\end{lemma}
From Lemma \ref{Lemma_typicalityCriterion}, we obtain conditional typicality and conditional divergence
lemmas for the proposed generalised typicality criterion.
\begin{lemma}
Given a bivariate Gaussian memoryless source $p_{XY}$ we set 
\begin{align*}
&\sett_\delta^{Y^n|x^n}=\{y^n \in \sety^n,\ (x^n,y^n) \in \sett_{\delta}^{X^n Y^n}\}, \quad x^n \in \setx^n\\
&\sett_\delta^{X^n|Y^n}=\{x^n \in \setx^n,\ p_{Y^n|X^n}(\sett_\delta^{Y^n|x^n}|x^n)>0\}. \end{align*}
Then
\begin{align}
    &\lim_{n\to \infty} p_{X^n}(\sett_{\delta}^{X^n|Y^n})=1, \label{eq:TypCondition}\\
    &\lim_{n\to \infty} p_{Y^n|X^n}(\sett_\delta^{Y^n|x^n})=1, \quad \forall x^n \in \sett_{\delta}^{X^n|Y^n}.
\end{align}
\label{lemma_cond_typicality}
\end{lemma}
\begin{proof}
$\sett_\delta^{X^nY^n}$ is $p_{X^nY^n}$-measurable because $i_{X^nY^n}$ is a measurable function. For i.i.d. sequence pairs $(X^n,Y^n)$, it follows from the Weak Law of Large Numbers (WLLN) that for any $\delta>0$
\begin{equation*}
\lim_{n\to \infty} \Pr\{|\frac{1}{n} i_{X^nY^n}(X^n,Y^n)- I(X;Y)|<\delta \}=1,
\end{equation*}
where $\mathbb{E}\left[\frac{1}{n} i_{X^nY^n}(X^n,Y^n)\right]=I(X;Y)$. Thus $\sett_\delta^{X^nY^n}$ satisfies condition \eqref{eq:jointTypicality}.
\end{proof}
\begin{lemma}
Given a bivariate Gaussian memoryless source $p_{XY}$, for all $n \in \mathbb{N}$ and $x^n \in \sett_{\delta}^{X^n|Y^n} $
\begin{align}
    & 2^{- n[{I}(X;Y)+\delta]} \leq p_{Y^n}(\sett_\delta^{Y^n|x^n}) \leq 2^{- n[{I}(X;Y)-\delta]} \\
    & 2^{- n[{I}(X;Y)+\delta]} \leq p_{X^n} p_{Y^n}(\sett_\delta^{Y^nX^n}) \leq 2^{-[ n({I}(X;Y)-\delta]},
\end{align}
where \begin{align*} p_{X^n} p_{Y^n}(\sett_\delta^{Y^nX^n})
=\int_{x^n \in \sett_\delta^{X^n|Y^n}} p_{Y^n}(\sett_\delta^{Y^n|x^n})dp_{X^n}(x^n).\end{align*}
\label{lemma_cond_divergence}
\end{lemma}
\begin{proof}
The proof is similar to the proof in \cite[Lemma 3]{Generalized_typ}.
\end{proof}
\section{System Model, Definitions and Main Result}
\label{systemmodelanddefinitions}
In this section, we introduce our system model and propose a single-letter characterization of the CR capacity for the scenario presented in Fig. \ref{fig:System}.
\subsection{System Model}
\label{systemmodel}
Let a bivariate Gaussian memoryless source $p_{XY}$ with two components, with  generic variables $X$ and $Y$ on alphabets $\mathcal{X} \subseteq \mathbb{R}$ and $\mathcal{Y}\subseteq \mathbb{R}$, correspondingly, be given.
The outputs of $X$ are observed only by Terminal $A$ and those of $Y$ only by Terminal $B$. Both outputs have length $n.$ We further assume that the joint distribution of $(X,Y)$ is known to both terminals. Terminal $A$
can send information to Terminal $B$ over a memoryless channel $W.$ The Shannon capacity of the channel $W$ is denoted by $C(W)$. There are no other resources available to any of the terminals.  \\
A CR-generation protocol \cite{part2} of block length $n$ consists of:
\begin{enumerate}
    \item a function $\Phi$ that maps $X^n$ into a random variable $K$ with alphabet $\mathcal{K}$ generated by Terminal $A$,
    \item a function $\Lambda$ that maps $X^n$ into the input sequence $T^n$
    \item a function $\Psi$ that maps $Y^n$ and the output sequence $Z^n$ into a random variable $L$ with alphabet $\mathcal{K}$ generated by Terminal $B$.
\end{enumerate}
This protocol generates a pair of random variable $(K,L)$ that is called permissible \cite{part2} if $K$ and $L$ are functions of the resources available at Terminal $A$ and Terminal $B$, respectively.
\begin{equation}
    K=\Phi(X^{n}), \ \     L=\Psi(Y^{n},Z^{n}).
    \label{KLSISOcorrelated}
\end{equation}
The system model is depicted in Fig. \ref{fig:System}.
\begin{figure}[htb!]
\centering
\tikzstyle{block} = [draw, rectangle, rounded corners,
minimum height=2em, minimum width=2cm]
\tikzstyle{blockchannel} = [draw, top color=white, bottom color=white!80!gray, rectangle, rounded corners,
minimum height=1cm, minimum width=.3cm]
\tikzstyle{input} = [coordinate]
\usetikzlibrary{arrows}
\scalebox{.9}{
\begin{tikzpicture}[scale= 1,font=\footnotesize]
\node[blockchannel] (source) {$P_{XY}$};
\node[blockchannel, below=2.4cm of source](channel) { memoryless channel};
\node[block, below left=2.4cm of source] (x) {Terminal $A$};
\node[block, below right=2.4cm of source] (y) {Terminal $B$};
\node[above=1cm of x] (k) {$K=\Phi(X^n)$};
\node[above=1cm of y] (l) {$L=\Psi(Y^n,Z^n)$};

\draw[->,thick] (source) -- node[above] {$X^n$} (x);
\draw[->, thick] (source) -- node[above] {$Y^n$} (y);
\draw [->, thick] (x) |- node[below right] {$T^n=\Lambda(X^n)$} (channel);
\draw[<-, thick] (y) |- node[below left] {$Z^n$} (channel);
\draw[->] (x) -- (k);
\draw[->] (y) -- (l);

\end{tikzpicture}}
\caption{Bivariate Gaussian memoryless source model with one-way communication over a memoryless channel}
\label{fig:System}
\end{figure}
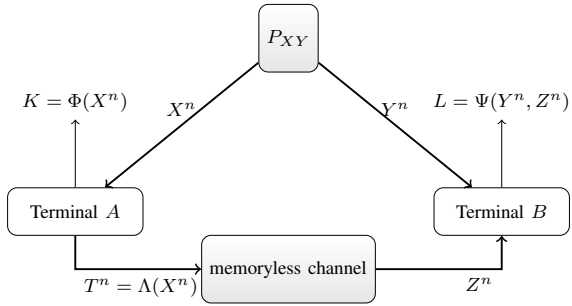
\subsection{Definitions and Main Result}
In this section, we provide the definition of an achievable CR rate and present the main result of the paper.
\begin{definition}  A number $H$ is called an achievable CR rate if there exists a non-negative constant $c$ such that for every $\epsilon>0$ and $\gamma>0$ and for sufficiently large $n$ there exists a permissible  pair of random variables $(K,L)$ such that
\begin{equation}
    \Pr\{K\neq L\}\leq \epsilon, 
    \label{errorcorrelated}
\end{equation}
\begin{equation}
    |\mathcal{K}|\leq 2^{cn},
    \label{cardinalitycorrelated}
\end{equation}
\begin{equation}
    \frac{1}{n}H(K)> H-\gamma.
     \label{ratecorrelated}
\end{equation}
\end{definition}
\begin{definition} 
The CR capacity $C_{CR}(p_{XY},W)$ is the maximum achievable CR rate.
\end{definition}
\begin{theorem}
For the model in Fig \ref{fig:System}, the CR capacity $C_{CR}(p_{XY},W)$ is equal to
\begin{equation}
    C_{CR}(p_{XY},W)= \underset{\substack{U \\{\substack{U \circlearrow{X} \circlearrow{Y}\\ I(U;X)-I(U;Y) \leq C(W)}}}}{\max} I(U;X).
    \label{eq:capacity}
\end{equation}
\label{main theorem}
\end{theorem}
In contrast to the discrete case in \cite{part2,globecom}, note that the CR capacity can reach infinity \cite{cover}. If the $I(X,Y)=+\infty$, then the single-letter characterization in \eqref{eq:capacity} can be reduced to the following form:
\begin{align*}
  C_{CR}(p_{XY},W) &= \underset{\substack{U \\{\substack{U \circlearrow{X} \circlearrow{Y}}}}}{\max} I(U;X)  = +\infty.
\end{align*}
If $I(X,Y)=+\infty$, then $Y$ is a linear function of $X$ with probability one \cite{correlationreference}. This implies that $I(U;X)-I(U;Y)=0$ with probability one.
\begin{remark}In our model, we distinguish two sources of randomness. The first one is obtained from the correlated source $p_{XY}$ and the second one by communicating over the channel $W$.
When the two continuous random variables $X$ and $Y$ are perfectly correlated, we can achieve infinite CR capacity without communicating over the channel, since the joint distribution of $(X,Y)$ is known to both terminals.
\end{remark}

\section{Direct Proof of Theorem \ref{main theorem}}
\label{direct}
In this section, we provide the direct proof of Theorem \ref{main theorem}. We distinguish two cases. The first one is when the $X$ and $Y$ are perfectly correlated, i.e., the mutual information $I(X;Y)$ is infinite. The second one is when $I(X;Y)$ is finite. In the latter case, we can use the typicality criteria presented in Section \ref{subsec: typicality}.
\subsection{$I(X,Y)$ is Infinite} 
We recall that $p_{XY}$ is a bivariate Gaussian source. The mutual information $I(X,Y)$ is given by
\begin{equation*}
    I(X,Y)=-\frac{1}{2}\log(1-\rho^2),
\end{equation*}
where $\rho$ is the correlation coefficient between $X$ and $Y$. That means $I(X,Y)=+\infty$ iff $|\rho|=1$, i.e., $X$ and $Y$ are perfectly correlated. In such a situation, $Y$ is a linear function of $X$ with probability one \cite{correlationreference}. We set
\begin{equation*}
    Y=g(X),
\end{equation*}
where $g \colon \setx \longrightarrow \sety$ is a linear function.
Therefore, almost surely,
we do not need to communicate over the channel. Since $X$ and $Y$ are perfectly correlated, we can achieve infinite CR capacity without sending any information over the channel. We prove that it is sufficient that the terminals $A$ and $B$ observe one symbol $X$ and $Y$, respectively. In the following, we first prove the existence of a function $\Phi$ that converts the Gaussian RV $X$ to the RV $K$ uniformly distributed on $\mathcal{K}=\{1,2,\ldots, |\mathcal{K}|\}$. It is worth noting that we do not pay any price for the uniformity. We can convert a random experiment with a Gaussian distribution to another one with uniform distribution with zero error probability.
 \begin{lemma}
Assume $X$ has a normal distribution with mean $\mu_X$ and variance $\sigma^2_X>0$. We denote by $F$ the cumulative distribution function of the standard normal distribution.
Let for $\sigma^2>0$ the RV $\tilde{X}$ be defined as $\tilde{X}=F(\frac{X-\mu_X}{\sqrt{\sigma^2_X}})$. $\tilde{X}$ is uniformly distributed on $(0,1)$.
\label{lemma_conversion}
\end{lemma} 
The proof of Lemma \ref{lemma_conversion} is analogous to the proof of \cite[Lemma 7]{isit_paper}. We then discretize $\tilde{X}$ using the function $d$ as described in \cite{isit_paper}. 
\begin{align*}
    d& \colon (0,1) \longrightarrow \{1,2,\ldots,|\mathcal{K}|\} \\
    & \colon \tilde{x} \mapsto k,\quad k \in \mathcal{K}.
\end{align*}
We set 
\begin{align*}
\Phi & \colon \mathbb{R}^n\longrightarrow \{1,2,\ldots, |\mathcal{K}|\}, \\
    & \colon x\mapsto {d\circ F}(x).
\end{align*}
 We set $|\mathcal{K}|=2^{nc},\quad c>0$. Thus condition \eqref{cardinalitycorrelated} is satisfied. Let $\Psi=\Phi \circ g^{-1}$. If $K=\Phi(X)$, then
\begin{align*}
    L& =\Psi(Y)\\
    & ={\Phi \circ g^{-1}}(g(X))\\
    & =K.
\end{align*}
Thus, \eqref{errorcorrelated} is satisfied. Now, we want to compute the entropy of $K$.
\begin{align*}
H(K)&= \log({|\mathcal{K}|})\\
&=nc,\quad c>0\\
&=n H.
\end{align*}
Since the constant $c$ can be chosen arbitrarily, then  \eqref{ratecorrelated} is satisfied for any positive $H.$  Thus, we have proved that any CR rate is achievable. This implies that the CR capacity is infinite in this case. This completes the proof.
\subsection{$I(X,Y)$ is Finite}
We consider the same code construction as used in \cite{part2} based on the same type of binning as for the Wyner-Ziv problem. Let $\epsilon,\gamma>0.$
Let $U$ be an arbitrary random variable on $\setu$ satisfying $U \circlearrow{X} \circlearrow{Y}$ and  $I(U;X)-I(U;Y)< C(W)$.
We are going to show that $H=I(U;X)$ is an achievable CR rate.

Let $p_{U|X}$ be a “channel” from $X$ to $U$. \\
{\bf{Code Construction}}: We generate $N_1N_2$ codewords $u^n(i,j),\quad i=1,\ldots,N_1,\ j=1,\ldots,N_2$ by choosing the $n.(N_1N_2)$ symbols $u_l(i,j)$ independently at random using $p_U$ (computed from $p_{XU}$). Each realization $u^n_{i,j}$ of $U^n_{i,j}$ is known to both terminals. For some $\delta>0$, let
{{\begin{align}
N_{1}&=2^{\left(n[I(U;X)-I(U;Y)+4\delta]\right)} \nonumber\\
N_{2}&=2^{\left(n[I(U;Y)-2\delta]\right)}. \nonumber
\end{align}}}
{\bf{Encoder}}: Let $(x^n,y^n)$ be any realization of $(X^n,Y^n)$. Given $x^n$ with $(x^n,y^n) \in \sett_\delta^{X^nY^n}$, try to find a pair $(i,j)$ such that $\left(x^n,u^n(i,j) \right) \in \sett_{\delta}^{X^nU^n}$ and $u^n(i,j) \in \sett_\delta^{U^n|Y^n}$. If successful, let $f(x^n)=i$. If no such ${u}^n(i,j)$ exists, then $f({x}^n)=N_1+1$ and $\Phi({x}^n)$ is set to a constant sequence ${u}^n_0$ different from all the ${u}^n(i,j)$s and known to both terminals. We choose $\delta$ to be sufficiently small such that
 \begin{align}
     \frac{\log \lVert f \rVert}{n}&=\frac{\log(N_1+1)}{n} \nonumber \\
     &\leq C(W)-\delta', \ \delta'>0,
     \label{inequalitylogfSISO}
      \end{align}
where $\lVert f \rVert$ refers to the cardinality of the set of messages $\{{i}^{\star}=f({x}^n)\}$. The message $i^{\star}=f({x}^n)$, with $i^{\star}\in\{1,\hdots,N_1+1\}$, is encoded to a sequence ${t}^n$ using a suitable \textit{forward error correcting code} with rate $\frac{\log \lVert f \rVert}{n}$ satisfying \eqref{inequalitylogfSISO} and with error probability not exceeding $\frac{\epsilon}{2}$ for sufficiently large $n$. The sequence ${t}^n$ is sent over the channel $W$. \\
{\bf Decoder}: Let ${z}^n$ be the channel output sequence. Terminal $B$ decodes the message $\hat{i}^{\star}$ from the knowledge of ${z}^n$. Given $\hat{i}^{\star}$ and $y^n$, try to find $\tilde{j}$ such that $\left(y^n,u^n(\hat{i}^{\star},\tilde{j})\right) \in \sett_{\delta}^{Y^nU^n}$. If successful, let $L({y}^n,\hat{i}^{\star})={u}^n(\hat{i}^{\star},\tilde{j})$. If there is no such ${u}^n(\hat{i}^{\star},\tilde{j})$ or there are several, $L$ is set to ${u}^n_0$ (since $K$ and $L$ must have the same alphabet).\\
{\bf Error Analysis}: We consider the following error events. \\
\textbullet \ $\sete_1:=\left\{ (X^n, Y^n) \notin \sett_\delta^{X^nY^n} \right\}$.\\
   \textbullet \ Suppose that $(x^n,y^n) \in \sett_\delta^{X^nY^n}$ but the encoder cannot find a pair $(i,j)$ such that $\left(x^n,u^n(i,j)\right) \in \sett_\delta^{X^nU^n} $ and $u^n \in \sett_\delta^{U^n|Y^n}$,\\
  $ \sete_2:= \bigcap_{\substack{i=1,\ldots,N_1 \\j=1,\ldots,N_2}} \big\{ \left(X^n,U^n(i,j)\right) \notin \sett_\delta^{X^nU^n}  \cup U^n(i,j) \notin \sett_\delta^{U^n|Y^n} \big\}$. \\
   \textbullet \  Suppose that $(x^n,y^n) \in \sett_\delta^{X^nY^n}$ and the encoder finds a pair $(i,j)$ such that $\left(x^n,u^n(i,j)\right) \in \sett_\delta^{X^nU^n} $ with $u^n(i,j) \in \sett_{\delta}^{U^n|Y^n}$. However, the decoder finds $\tilde{j}\neq j$ such that 
    $\left(y^n,u^n(\hat{i},\tilde{j})\right) \in \sett_\delta^{Y^nU^n} $,\\
    $\sete_3:=\cup_{\substack{\tilde{j}=1,\ldots,N_2 \\ \tilde{j}\neq j}} \left\{ \left(Y^n,U^n(\hat{i},\tilde{j})\right) \in \sett_\delta^{Y^nU^n} \right\}$.\\
    \textbullet \ Suppose that $(x^n,y^n) \in \sett_\delta^{X^nY^n}$ and the encoder finds a pair $(i,j)$ such that $\left(x^n,u^n(i,j)\right) \in \sett_\delta^{X^nU^n} $ with $u^n(i,j) \in \sett_{\delta}^{U^n|Y^n}$. However, the decoder cannot find ${j}$ such that 
    $\left(y^n,u^n(\hat{i},{j})\right) \in \sett_\delta^{Y^nU^n} $, \\
    $\sete_4:= \bigg\{\cap_{j=1,\ldots,N_2}\left\{ \left(Y^n,U^n(\hat{i},{j})\right) \notin \sett_\delta^{Y^nU^n} \right\}\bigg\} \bigcap \sete_2^c$.

     We denote by $P_e$ the probability of the overall error event. It follows from the union bound that 
 \begin{equation*}
  P_e \leq \Pr\{\sete_1\}+\Pr\{\sete_2\}+\Pr\{\sete_3\}+\Pr\{\sete_4\}.\end{equation*}
In the following, we compute an upper-bound on the overall error probability.
\begin{align*}
   \Pr\{\sete_1\}&=p_{XY}^n\left((\sett_\delta^{X^nY^n})^c \right) \\
   &=1-p_{XY}^n\left(\sett_\delta^{X^nY^n}\right)\\
    &\overset{(a)}{\leq}\beta_1(n),\quad \lim_{n\to\infty}\beta_1(n)=0,
\end{align*}
where $(a)$ follows from Lemma \ref{Lemma_typicalityCriterion} as $p_{XY}^n(\sett_\delta^{X^nY^n})$ satisfies condition \eqref{eq:TypCondition} w.r.t. the typicality criterion in \eqref{eq:criterion}.

\begin{align*}
    \Pr\{\sete_3\}& \overset{(a)}{\leq} \sum_{\tilde{j}\neq j} \Pr\left\{\left(Y^n,U^n(\hat{i},\tilde{j})\right) \in \sett_\delta^{Y^nU^n}\right \} \\
    & \overset{(b)}{<} N_2 \cdot 2^{-n(I(U,Y)+\delta)} \\
    & = 2^{-n \delta},\quad \beta_3(n):=2^{-n \delta},
\end{align*}
where $(a)$ follows from the union bound and $(b)$ follows from Lemma \ref{lemma_cond_divergence}. $p_{UY}$ can be computed from $p_{U|X}$ and $p_{XY}$.
\begin{align*}
    p_{UY}(u^n,y^n)&  =\int_{x^n \setx^n} p_{U|XY}^n(u^n|x^n,y^n) p_{XY}^n(x^n,y^n) dx^n \\
    & \overset{(a)}{=}\int_{x^n \setx^n} p_{U|X}^n(u^n|x^n,y^n) p_{XY}^n(x^n,y^n) dx^n.
\end{align*}
$(a)$ follows because $U\circlearrow X\circlearrow Y$ forms a Markov chain.
We compute an upper-bound for $\Pr\{\sete_4\}$.
\begin{align*}
    \Pr\{\sete_4 \}& =
\Pr \bigg\{\cap_{j=1,\ldots,N_2}\left\{ \left(Y^n,U^n(\hat{i},\tilde{j})\right) \notin \sett_\delta^{Y^nU^n} \right\} \\
&\quad \bigcap \sete_2^c\bigg\} \\
    &\leq  \Pr\bigg\{ \bigcap_{j=1,\ldots,N_2} \Big\{\left(Y^n,U^n(\hat{i},{j})\right) \notin \sett_\delta^{Y^nU^n} \\
    &\quad \cap U^n(\hat{i},{j}) \in \sett_\delta^{U^n|Y^n}\Big\} \bigg\} \\
    &\leq \beta_4(n),\quad \lim_{n \to \infty} \beta_4(n)=0.
\end{align*}
Now, we compute an upper-bound for $\Pr\{\sete_2\}$.
\begin{align*}
&\Pr\{\sete_2\}\\
    &=\int_{x^n \in \setx^n} p_{X^n}(x^n)\Pr\{\sete_2|X^n=x^n\} dx^n \\
    &=\int_{\substack{x^n \notin \sett_\delta^{X^n|U^n}}} p_{X^n}(x^n) \Pr\{\sete_2|X^n =x^n\} dx^n \\
    &+\int_{\substack{x^n \in \sett_\delta^{X^n|U^n}}} \Pr\bigg\{\bigcap_{\substack{i=1,\ldots,N_1 \\j=1,\ldots,N_2}} \left(x^n,U^n(i,j)\right) \notin \sett_\delta^{X^n U^n}\\ & \quad \cup U^n(i,j)\notin \sett_\delta^{U^n|Y^n}|X^n=x^n\bigg\}p_{X^n}(x^n)dx^n \\
    & {\leq} p_{X}^n\left((\sett_\delta^{X^n|U^n})^c\right)\\
    &+ \int_{\substack{x^n \in \sett_\delta^{X^n|U^n}}} \Pr\bigg\{\bigcap_{\substack{i=1,\ldots,N_1 \\j=1,\ldots,N_2}} U^n(i,j) \notin \sett_\delta^{U^n|X^n}\\ 
    & \quad \cup U^n(i,j)\notin \sett_\delta^{U^n|Y^n}|X^n=x^n\bigg\}p_{X^n}(x^n)dx^n \\
    &\overset{(a)}{\leq} \beta(n) + \int_{\substack{x^n \in \sett_\delta^{X^n|U^n}}} p_{X^n}(x^n) \prod_{\substack{i=1,\ldots,N_1 \\j=1,\ldots,N_2}} \bigg( \Pr\big\{U^n(i,j) \notin \\
    & \sett_\delta^{U^n|X^n}|X^n=x^n\big\} + \Pr\big\{U^n(i,j) 
    \notin \sett_\delta^{U^n|Y^n}|X^n=x^n\big\}\bigg) dx^n\\
    &\overset{(b)}{\leq} \beta(n)+ \int_{\substack{x^n \in \sett_\delta^{X^n|U^n}}}  \left(1-2^{-n(I(U,X)+\delta)}+\beta'(n)\right)^{N_1N_2} \\
    &\quad p_{X^n}(x^n) dx^n \\
    & \overset{(c)}{\leq} \beta(n)+  \exp\left(-2^{n(-I(U,X)-\delta)}-\beta'(n)\right)^{N_1N_2}\\
    & \leq \beta(n)+ \exp(\left(-2^{n(-I(U,X)-\delta)}\right)^{N_1N_2}\cdot \exp(-\beta'(n))^{N_1N_2}\\
    &\leq \beta_2(n), \quad \lim_{n\to\infty}\beta_2(n)\overset{(d)}{=}0,
\end{align*}
where $(a)$ follows because the $N_1N_2$ events of the intersection are independent and from Lemma \ref{lemma_cond_typicality}, $(b)$ follows from Lemma \ref{lemma_cond_typicality} and Lemma \ref{lemma_cond_divergence} with $\lim_{n\to \infty}\beta'(n)=0$, $(c)$ follows because $(1-x)^m\leq \exp(-mx)$ and $(d)$ follows because $\lim_{n\to\infty}\beta(n)=0$ and $\frac{1}{n}\log(N_1N_2)>I(U,X)$.
Therefore, for large sufficiently $n$
\begin{equation*}
  P_e \leq \sum_{i=1}^{4} \beta_i(n) \leq \frac{\epsilon}{2}.
\end{equation*}
Now, we are going to show that  $(K,L)$ satisfies \eqref{errorcorrelated}, \eqref{cardinalitycorrelated} and \eqref{ratecorrelated}.
Clearly, (\ref{cardinalitycorrelated}) is satisfied  for $c=2(H(X)+1)$, $n$ sufficiently large:
{{\begin{align}
|\mathcal{K}|&=N_1 N_2+1 \nonumber \\
             &= 2^{(n\left[I(U;X)+\delta\right])}+1 \nonumber \\
             &\leq 2^{(2n\left[I(U;X)+\delta \right])}.\nonumber
\end{align}}}
For a fixed $u^n(i,j) \in \setu^n$, we compute the following probability.
\begin{align*}
    &\Pr\{K={u}^n(i,j)\} \\&=\int_{{x}^n\in\sett_\delta^{X^n|U^n}}\Pr\{K={u}^n(i,j)|X^n={x}^n\}p_{X}^n({x}^n) dx^n \\
&\quad+\int_{{x}^n\in(\sett_\delta^{X^n|U^n})^c}\Pr\{K={u}^n(i,j)|X^n={x}^n\}p_{X}^n({x}^n) dx^n\\
&\overset{(a)}{=}\int_{{x}^n\in \sett_\delta^{X^n|U^n}}\Pr\{K={u}^n(i,j)|X^n={x}^n\}p_{X}^n({x}^n) dx^n  \\
&\leq \int_{{x}^n\in\sett_\delta^{X^n|U^n}}p_{X}^n({x}^n) dx^n =p_{X}^{n}(\sett_\delta^{X^n|U^n})\\
& \overset{(b)}{\leq} 2^{\left(-n(I(U;X)+\delta)\right)},
\end{align*}
where $(a)$ follows because for $(x^n,{u}^n(i,j))$ being not jointly typical, we have $\Pr \{K={u}^n(i,j)|X^n={x}^n\}=0$ and $(b)$ follows from Lemma \ref{lemma_cond_divergence}. This yields
{{\begin{align}
H(K) & \geq 
n (I(U;X)+ \delta) \nonumber\\
& = n H+o(n). \nonumber
\end{align}}}Thus, (\ref{ratecorrelated}) is satisfied. Now, it remains to prove that \eqref{errorcorrelated} is satisfied. We further define $I^\star=f(X^n)$ to be the random variable modeling the message encoded by Terminal $A$ and $\hat{I}^\star$ to be the random variable modeling the message decoded by Terminal $B$. 
We have:
\begin{align}
    \Pr\{K\neq L\} \nonumber &=\Pr\{K\neq L|I^\star=\hat{I}^\star\}\Pr\{I^\star=\hat{I}^\star\} \nonumber \\
        &\quad + \Pr\{K\neq L|I^\star\neq \hat{I}^\star\}\Pr\{I^\star\neq\hat{I}^\star\} \nonumber \\
        &\leq \Pr\{K\neq L|I^\star=\hat{I}^\star\}+ \Pr\{I^\star\neq\hat{I}^\star\}\nonumber.
\end{align}
we define the following event:
\begin{align}
    \sete= ``K(X^n) \ \text{is equal to none of the} \  {u}^n(i,j)s". \nonumber
\end{align}
We have
\begin{align}
   & \Pr\{K\neq L|I^\star=\hat{I}^\star\} \nonumber \\
   &= \Pr\{K\neq L|I^\star=\hat{I}^\star,\sete]\Pr\{\sete|I^\star=\hat{I}^\star\}   \nonumber \\
   &\quad + \Pr\{K\neq L|I^\star=\hat{I}^\star,\sete^c\}\Pr\{\sete^c|I^\star=\hat{I}^\star\} \nonumber \\
   &\overset{(a)}{=}\Pr\{K\neq L|I^\star=\hat{I}^\star,\sete^c\}\Pr\{\sete^c|I^\star=\hat{I}^\star\} \nonumber \\
   &\leq \Pr\{K\neq L|I^\star=\hat{I}^\star,\sete^c\},\nonumber
\end{align}
where $(a)$ follows from $\Pr\{K\neq L|I^\star=\hat{I}^\star,\sete\}=0,$ since conditioned on $I^\star=\hat{I}^\star$ and $\sete$, we know that $K$ and $L$ are both equal to $u^n_0$.
It follows that
\begin{align}
    &\Pr\{K\neq L\} \nonumber \\
    &\leq \Pr\{K\neq L|I^\star=\hat{I}^\star,\sete^c\}+ \Pr\{I^\star\neq\hat{I}^\star] \nonumber \\
    &\leq \Pr\{\cup_{i=1}^{4}\sete_i\}+ \Pr\{I^\star\neq\hat{I}^\star\} \nonumber \\
    &\overset{(a)}{\leq} P_e+\frac{\epsilon}{2}\\
    &\leq \epsilon,
\end{align}
 where $(a)$ follows from the union bound.\\
This completes the direct proof.
\section{Converse proof of  Theorem \ref{main theorem}}
\label{converse}
Let $(K,L)$ be a permissible pair according to a fixed CR-generation protocol of block-length $n,$ as introduced in Section \ref{systemmodel}.
We further assume that $(K,L)$ satisfies \eqref{errorcorrelated}
 \eqref{cardinalitycorrelated} and \eqref{ratecorrelated}.
We are going to show for some $\epsilon'(n)>0$ that
\begin{align}
    \frac{H(K)}{n} \leq \underset{ \substack{U \\{\substack{U \circlearrow{X} \circlearrow{Y}\\ I(U;X)-I(U;Y) \leq C(W)+\epsilon'(n)}}}}{\max} I(U;X), \nonumber 
\end{align}
where $\underset{n\rightarrow \infty}{\lim}\epsilon'(n)$ can be made arbitrarily small for $\epsilon>0$ chosen arbitrarily small.
In our proof, we will use  the following lemma: 
\begin{lemma} (Lemma 17.12 in \cite{IT_CiKo})
For arbitrary random variables $S$ and $R$ and sequences of random variables $X^{n}$ and $Y^{n}$, it holds that
\begin{align}
 &I(S;X^{n}|R)-I(S;Y^{n}|R) \nonumber \\ 
 &=\sum_{i=1}^{n} I(S;X_{i}|X_{1},\dots, X_{i-1}, Y_{i+1},\dots, Y_{n},R) \nonumber \\ &\quad -\sum_{i=1}^{n} I(S;Y_{i}|X_{1},\dots, X_{i-1}, Y_{i+1},\dots, Y_{n},R) \nonumber \\
 &=n[I(S;X_{J}|V)-I(S;Y_{J}|V)],\nonumber
\end{align}
where $V=(X_{1},\dots, X_{J-1},Y_{J+1},\dots, Y_{n},R,J)$, with $J$ being a random variable independent of $R$,\ $S$, \ $X^{n}$ \ and $Y^{n}$ and uniformly distributed on $\{1 ,\dots, n \}$.
\label{lemma1}
\end{lemma}
Let $J$ be a random variable uniformly distributed on $\{1,\dots, n\}$ and independent of $K$, $X^n$ and $Y^n$. We further define $U=(K,X_{1},\dots, X_{J-1},Y_{J+1},\dots, Y_{n},J).$ It holds that $U \circlearrow{X_J} \circlearrow{Y_J}.$ \\
Notice  that
{{\begin{align}
H(K)&\overset{(a)}{=}H(K)-H(K|X^{n})\nonumber\\
&=I(K;X^{n}) \nonumber\\
&\overset{(b)}{=}\sum_{i=1}^{n} I(K;X_{i}|X_{1},\dots, X_{i-1}) \nonumber\\
&=n I(K;X_{J}|X_{1},\dots, X_{J-1},J) \nonumber\\
&\overset{(c)}{\leq }n I(U;X_{J}), \nonumber
\end{align}}} where$(a)$ follows because $K=\Phi(X^n)$ and $(b)$ and $(c)$ follow from the chain rule for mutual information.
Applying Lemma \ref{lemma1} for $S=K$, $R=\varnothing$ with $V=(X_1,\hdots, X_{J-1},Y_{J+1},\hdots, Y_{n},J)$ yields
\begin{align}
&I(K;X^{n})-I(K;Y^{n}) 
\nonumber \\&=n[I(K;X_{J}|V)-I(K;Y_{J}|V)] \nonumber\\
&\overset{(a)}{=}n[I(KV;X_{J})-I(K;V)-I(KV;Y_{J})+I(K;V)] \nonumber\\
&\overset{(b)}{=}n[I(U;X_{J})-I(U;Y_{J})], 
\label{UhilfsvariableMIMO1}
\end{align}
where $(a)$ follows from the chain rule for mutual information and $(b)$ follows from $U=(K,V)$. \\
It results using (\ref{UhilfsvariableMIMO1}) that
\begin{align}
n[I(U;X_{J})-I(U;Y_{J})]
&=I(K;X^{n})-I(K;Y^{n}) \nonumber\\
&=H(K)-I(K;Y^{n})\nonumber \\ 
&=H(K|Y^n).
\label{star2MIMO2}
\end{align}
Next, we will show for some $\epsilon'(n)>0$ that
\begin{align}
    \frac{H(K|Y^n)}{n}\leq C(W)+\epsilon'(n). \nonumber
\end{align}
We have
\begin{equation}
H(K|Y^{n})=I(K;Z^{n}|Y^{n})+H(K|Y^{n}Z^{n}).\label{boxed2}
\end{equation}
On the one hand, it holds that
\begin{align} 
 I(K;Z^{n}|Y^{n})&\leq I(X^{n}K;Z^{n}|Y^{n}) \nonumber\\
& \overset{(a)}{\leq }I(T^n;Z^n|Y^{n})  \nonumber \\
& =  h(Z^n|Y^{n})- h(Z^n|T^n,Y^{n}) \nonumber \\
& \overset{(b)}{=}  h(Z^n|Y^{n})- h(Z^n|T^n) \nonumber \\
& \overset{(c)}{\leq }   h(Z^n)- h(Z^n|T^n) \nonumber \\
& = I(T^n;Z^n)  \nonumber \\
& \overset{(d)}{=} \sum_{i=1}^{n} I(Z_{i};T^n|Z^{i-1}) \nonumber \\
& = \sum_{i=1}^{n} h(Z_{i}|Z^{i-1})-h(Z_{i}|T^n,Z^{i-1}) \nonumber \\
& \overset{(e)}{=} \sum_{i=1}^{n} h(Z_{i}|Z^{i-1})-h(Z_{i}|T_{i}) \nonumber \\
& \overset{(f)}{\leq} \sum_{i=1}^{n} h(Z_{i})-h(Z_{i}|T_{i}) \nonumber \\
&=\sum_{i=1}^{n} I(T_{i};Z_{i}) \nonumber \\
&\leq n C(W), \label{part1}
\end{align}
where $(a)$ follows from the Data Processing Inequality because $Y^{n}\circlearrow{X^{n}K}\circlearrow{T^n}\circlearrow{Z^{n}}$ forms a Markov chain, where we used the fact that the Data Processing inequality holds also for continuous random variables \cite{dataprocessing}, $(b)$ follows because $Y^{n}\circlearrow{X^{n}K}\circlearrow{T^n}\circlearrow{Z^{n}}$ forms a Markov chain, $(c)(f)$ follow because conditioning does not increase entropy, $(d)$ follows from the chain rule for mutual information and $(e)$ follows because $T_{1},\dots, T_{i-1},T_{i+1},\dots, T_{n},Z^{i-1} \circlearrow{T_{i}}\circlearrow{Z_{i}}$ forms a Markov chain. 
\color{black}
On the other hand, it holds that
\begin{align}
H(K|Y^{n},Z^{n})&\overset{(a)}{\leq } H(K|L) \nonumber \\
&\overset{(b)}{\leq } 1+\log\lvert \mathcal{K} \rvert \Pr[K\neq L] \nonumber \\
&\overset{(c)}{\leq }1+\epsilon c n, \label{part2}
\end{align}
where (a) follows from $L=\Psi(Y^{n},Z^{n})$ in \eqref{KLSISOcorrelated}, 
(b) follows from Fano's Inequality using \eqref{errorcorrelated} and (c) follows from \eqref{cardinalitycorrelated}.

It follows from \eqref{boxed2}, \eqref{part1} and \eqref{part2} that
\begin{align}
     \frac{H(K|Y^n)}{n}\leq C(W)+\epsilon'(n),
\end{align}
where $\epsilon'(n)=\frac{1}{n}+\epsilon c.$
From \eqref{star2MIMO2}, we deduce that 
\begin{align}
    I(U;X_{J})-I(U;Y_{J})\leq C(W)+\epsilon'(n).
\end{align}
Since the joint distribution of $X_{J}$ and $Y_{J}$ is equal to $p_{XY}$, $\frac{H(K)}{n}$ is upper-bounded by $I(U;X)$ subject to $I(U;X)-I(U;Y) \leq C(W) + \epsilon'(n)$ with $U$ satisfying $U \circlearrow{X} \circlearrow{Y}$. As a result, it holds that
\begin{align}
    \frac{H(K)}{n} \leq \underset{ \substack{U \\{\substack{U \circlearrow{X} \circlearrow{Y}\\ I(U;X)-I(U;Y) \leq C(W)+\epsilon'(n)}}}}{\max} I(U;X).
    \nonumber
\end{align}
Here, $\underset{n\rightarrow \infty}{\lim}\epsilon'(n)$ can be made arbitrarily small by choosing $\epsilon$ to be an arbitrarily small positive constant.
This completes the converse proof of Thereom \ref{main theorem}. 
\section{conclusion}
In this paper, we investigated the problem of CR generation from correlated Gaussian sources with communication over noisy channels. We extended the CR capacity formula established in \cite{part2} to Gaussian sources and showed that in contrast to the discrete case, where the CR capacity is always finite, one can achieve an infinite CR rate when the Gaussian sources are perfectly correlated. The obtained results are highly useful in the problem of correlation-assisted identification over Gaussian channels as well as the problem of identification over Gaussian channels in the presence of noisy feedback.
\section{Acknowledgments}
H.\ Boche was supported by the Deutsche Forschungsgemeinschaft (DFG, German
Research Foundation) within the Gottfried Wilhelm Leibniz Prize under Grant BO 1734/20-1, and within Germany’s Excellence Strategy EXC-2111—390814868 and EXC-2092 CASA-390781972. C.\ Deppe was supported in part by the German Federal Ministry of Education and Research (BMBF) under Grant 16KIS1005. 
H. Boche, W. Labidi and R. Ezzine were supported by the German Federal Ministry of Education and Research (BMBF) under Grant 16KIS1003K.
\label{conclusion}
\bibliographystyle{IEEEtran}
\bibliography{definitions,references}

\IEEEtriggeratref{4}



\end{document}